\documentclass[11pt]{article}
\usepackage[margin=1in]{geometry}
\usepackage{fullpage,wrapfig}
\usepackage{amsmath, amsthm, amssymb, algorithm,thmtools,algpseudocode,enumerate,caption,framed}
\usepackage{paralist, sidecap}
\usepackage[usenames,dvipsnames,svgnames,table]{xcolor}
\definecolor{darkgreen}{rgb}{0.0,0,0.9}
\usepackage[colorlinks=true,pdfpagemode=UseNone,citecolor=Blue,linkcolor=Red,urlcolor=Black,pagebackref]{hyperref}
\usepackage{tikz}
\usepackage[T1]{fontenc}
\usepackage[utf8]{inputenc}
\usepackage{authblk}
\usepackage{mdframed}
\usepackage{multirow}
\usepackage{diagbox}
\usepackage{graphicx,xcolor}

\newtheorem{theorem}{Theorem}[section]

\newtheorem{lemma}{Lemma}[section]

\newtheorem{obs}{Observation}[section]

\newcommand{\g}{\mathcal{R}}
\newcommand{\ml}{\mathcal{B}}
\newcommand{\np}{\mathsf{NP}}
\newcommand{\apx}{\mathsf{APX}}
\newcommand{\opt}{\mathsf{OPT}}
\newcommand{\ptas}{\mathsf{PTAS}}

\title{Constrained Orthogonal Segment Stabbing}

\author[1]{Sayan Bandyapadhyay}
\author[2]{Saeed Mehrabi}

\affil[1]{{\small Department of Computer Science, University of Iowa, Iowa City, USA.

\texttt{sayan-bandyapadhyay@uiowa.edu}}
}
\affil[2]{{\small School of Computer Science, Carleton University, Ottawa, Canada.

\texttt{saeed.mehrabi@carleton.ca}}
}

\date{}

\begin{document}

\maketitle

\begin{abstract}
Let $S$ and $D$ each be a set of orthogonal line segments in the plane. A line segment $s\in S$ \emph{stabs} a line segment $s'\in D$ if $s\cap s'\neq\emptyset$. It is known that the problem of stabbing the line segments in $D$ with the minimum number of line segments of $S$ is $\np$-hard~\cite{KatzMN05}. However, no better than $O(\log |S\cup D|)$-approximation is known for the problem.  In this paper, we introduce a constrained version of this problem in which every horizontal line segment of $S\cup D$ intersects a common vertical line. We study several versions of the problem, depending on which line segments are used for stabbing and which line segments must be stabbed. We obtain several $\np$-hardness and constant approximation results for these versions. Our finding implies, the problem remains $\np$-hard even under the extra assumption on input, but small constant approximation algorithms can be designed.
\end{abstract}

\section{Introduction}
\label{sec:introduction}
Let $S$ and $D$ be two sets of orthogonal line segments in the plane. In this paper, we study the \textit{orthogonal segment stabbing} problem, where the goal is to find a minimum-cardinality subset $S'\subseteq S$ such that every line segment in $D$ is \emph{stabbed} by at least one line segment in $S'$. A line segment $s\in D$ is stabbed by a line segment $s'\in S$ if and only if $s\cap s'\neq\emptyset$. Let $H$ and $V$ denote the set of input horizontal and vertical line segments, respectively, and $n=|H\sqcup V|$\footnote{Throughout the paper we use $\sqcup$ to denote disjoint union.}.

The orthogonal segment stabbing problem was studied by Katz et al.~\cite{KatzMN05} who proved that the problem is $\np$-hard even in the case when $S$ contains only vertical line segments and $D$ contains only horizontal line segments, i.e, $S=V$ and $D=H$. Notice that the problem is trivial when $S,D \subseteq V$ or $S,D \subseteq H$. Moreover, an $O(\log n)$-approximation algorithm for the problem is straightforward (by reducing the problem to set cover). To the best of our knowledge, no other approximation or inapproximabilty results are known for the general version of the problem. However, two special versions of the problem have been studied very recently. Bandyapadhyay and Basu Roy \cite{BandyapadhyayR17} studied a related art gallery problem, and from their work it follows that one can get a $\ptas$ for the orthogonal segment stabbing problem if $S\subseteq H$ or $S\subseteq V$. Mehrabi \cite{Mehrabi17} considered the version where no two horizontal (resp. vertical) line segments intersect each other (i.e., the intersection graph of the line segments in $S\cup D$ is bipartite, but $S,D \subseteq H\sqcup V$), and obtained a constant approximation.

\begin{center}
\begin{table}[t]
\centering
\begin{tabular}{ |p{1cm}|p{2cm}|p{2cm}|p{2.5cm}| }
 \hline
 \diagbox[height=8ex, width=3.7em]{\raisebox{0.3\height}{\enspace $S$}}{ \raisebox{-1\height}{\ $D$}}
 & $H$ & $V$ & $H\sqcup V$ \\ [1ex]
 \hline
\multirow{2}{*}{$H$} & \multirow{2}{*}{Polytime} & {$\np$-hard}  & {$\np$-hard} \\[1ex]
& (trivial) & {$\ptas$~\cite{BandyapadhyayR17}} & {$\ptas$~\cite{BandyapadhyayR17}}\\[1ex]
\hline
\multirow{2}{*}{$V$} & $\np$-hard$^{1}$ & \multirow{2}{*}{Polytime} & \multirow{2}{*}{$\ptas$~\cite{BandyapadhyayR17}}\\[1ex]
& {$\ptas$~\cite{BandyapadhyayR17}} & (trivial) &\\[1ex]
\hline
\multirow{3}{*}{$H\sqcup V$} & \multirow{2}{*}{5-approx.} &  \multirow{2}{*}{$\np$-hard } & {$\np$-hard} \\ [1ex]
&  &  & 7-approx.\\ [1ex]
& $\ptas$  & {2-approx.} & $(3+\epsilon)$-approx. \\ [1ex]
 \hline
\end{tabular}
\caption{A summary of our results for the $(S,D)$-stabbing problem. Each row corresponds to a set $S$ and each column corresponds to a set $D$. ($^1$ when each horizontal line segment must be intersected by exactly one selected line segment)}
\label{table:1}
\end{table}
\end{center}

In this paper, we introduce a constrained version of the orthogonal segment stabbing problem in which every horizontal line segment intersects a vertical line. More formally, for $S,D\subseteq H\sqcup V$, we define the \emph{$(S,D)$-stabbing problem} to be the problem of stabbing all the line segments in $D$ with the minimum number of line segments in $S$ with the constraint that all the line segments in $H$ intersect a vertical line $L_v$.
In this work, we study different versions of the $(S,D)$-stabbing problem. Considering these versions, we obtain the following results (see Table~\ref{table:1} for a summary of our results).
\begin{itemize}
\item The $(H,V)$-stabbing, $(H,H\sqcup V)$-stabbing, $(H\sqcup V,V)$-stabbing and $(H\sqcup V,H\sqcup V)$-stabbing problems are all $\np$-hard.
\item The $(V,H)$-stabbing problem is $\np$-hard when each horizontal line segment must be stabbed by exactly one line segment.
\item There exists an $O(n^6)$-time 5-approximation algorithm, and a local-search based $\ptas$ for the $(H\sqcup V,H)$-stabbing problem.
\item There exists an $O(n^5)$-time 2-approximation algorithm for the $(H\sqcup V,V)$-stabbing problem.
\item There exists an $O(n^6)$-time 7-approximation algorithm, and an $n^{O(\frac{1}{\epsilon^2})}$-time $(3+\epsilon)$-approximation algorithm (for any $\epsilon>0$) for the $(H\sqcup V,H\sqcup V)$-stabbing problem.
\end{itemize}

The $(S,D)$-stabbing problem is closely related to the minimum dominating set problem on axis-parallel polygonal chains in the plane. Among the previous work, the papers of Bandyapadhyay et al.~\cite{BandyapadhyayM018} and Chakraborty et al.~\cite{abs-1809-09990} are perhaps the most related to ours, as they consider the minimum dominating set problem on the intersection graph of ``L-shapes'' such that each L intersects a vertical line. The former showed that this problem is $\apx$-hard, and an 8-approximation algorithm for the problem was given by the latter. For more related work, see \cite{BandyapadhyayM018,abs-1809-09990} and the references therein.

\paragraph{Organization.} In Section~\ref{sec:prelimins}, we define some notations that we use throughout the paper, and
make a few observations that will be useful later. Then, we discuss our hardness results in Section~\ref{sec:npHardOne}, and the approximation results in Sections~\ref{sec:HUV,H} and~\ref{sec:approx}. 

\section{Preliminaries}
\label{sec:prelimins}
For a point $p$, let $x_p$ and $y_p$ be its $x$- and $y$-coordinates, respectively. For a horizontal line segment $h$ and a vertical line segment $v$ that intersect each other, let $I(h,v)$ denote their intersection point. The $y$-coordinate (resp., $x$-coordinate) of a horizontal (resp., vertical) line segment is defined to be the $y$-coordinate (resp., $x$-coordinate) of its endpoints. For a horizontal line segment $h$ with left endpoint $p$, call the line $y=y_p$, denoted by $L(h)$, the line-extension of $h$. Consider the line-extension $L(h)$. Let $V'$ be any subset of vertical line segments that intersect $L(h)$. For $v,v'\in V'$, $v$ and $v'$ are called consecutive w.r.t. $V'$ if there is no other line segment $v_1\in V'$ such that $I(L(h),v_1)$ lies in the open interval $(I(L(h),v),I(L(h),v'))$. We say a set of line segments $S_1$ \emph{hits} a set of line segments $S_2$ if for any line segment $s\in S_2$, there is a line segment $s'\in S_1$ such that $s'$ stabs $s$; i.e, $s\cap s'\ne \emptyset$.

Our $\ptas$ is based on the local search technique, which was introduced to computational geometry independently by Chan and Har-Peled~\cite{ChanH12}, and Mustafa and Ray~\cite{MustafaR10}. Consider a minimization problem in which the objective is to compute a feasible subset $A$ of a ground set $S$ whose cardinality is minimum over all such feasible subsets of $S$. Moreover, it is assumed that computing some initial feasible solution and determining whether a subset $A\subseteq S$ is a feasible solution can be done in polynomial time. The local search algorithm for a minimization problem is as follows. Fix some parameter $k$, and let $A$ be some initial feasible solution. Now, if there exist $A'\subseteq A$, $M\subseteq S\setminus A$ such that $|A'|\leq k$, $|M|<|A'|$ and $(A\setminus A')\cup M$ is a feasible solution, then we set $A=(A\setminus A')\cup M$. The above is repeated until no such ``local improvement'' is possible and we return $A$ as the final solution.

Let $\mathcal{B}$ and $\mathcal{R}$ be the solutions returned by the algorithm and an optimum solution, respectively. The following theorem establishes the connection between local search technique and obtaining a $\ptas$.
\begin{theorem}[\cite{ChanH12,MustafaR10}]
\label{thm:LSgivesPTAS}
Consider the solutions $\mathcal{B}$ and $\mathcal{R}$ for a minimization problem, and suppose that there exists a planar bipartite graph $H=(\mathcal{B}\cup\mathcal{R}, E)$ that satisfies the local exchange property: for any subset $\mathcal{B}'\subseteq\mathcal{B}$, $(\mathcal{B}\setminus\mathcal{B'})\cup N_H(\mathcal{B'})$ is a feasible solution, where $N_H(\mathcal{B'})$ denotes the set of neighbours of $\mathcal{B'}$ in $H$. Then, the local search algorithm yields a $\ptas$ for the problem.
\end{theorem}

The following simple observation will be useful in the next sections.
\begin{obs}\label{obs:merge}
Suppose $D=D_1\sqcup D_2$. If there is an $\alpha$-approximation algorithm $A_1$ for $(S,D_1)$-stabbing that runs in $f(n)$ time and a $\beta$-approximation algorithm $A_2$ for $(S,D_2)$-stabbing that runs in $g(n)$ time, then there is an $(\alpha+\beta)$-approximation algorithm $A$ for $(S,D)$-stabbing that runs in $f(n)+g(n)$ time.
\end{obs}
\begin{proof}
We show the existence of such an algorithm $A$ for the $(S,D)$-stabbing problem. Given an instance of the $(S,D)$-stabbing problem, consider an algorithm $A$ that uses $A_1$ to compute a solution for $(S,D_1)$-stabbing and uses $A_2$ to compute a solution for $(S,D_2)$-stabbing. Then, it returns the union of these two solutions as its solution. Clearly, the running time of $A$ is then $f(n)+g(n)$ as claimed. Moreover, since the size of an optimum solution for each of the $(S,D_1)$-stabbing and $(S,D_2)$-stabbing problems is at most the size of an optimum solution for $(S,D)$-stabbing and $D_1\cap D_2=\emptyset$, we get the desired approximation factor. 
\end{proof}



\section{Hardness Results}
\label{sec:npHardOne}
In this section, we first prove that the $(H\sqcup V,H\sqcup V)$-stabbing problem is $\np$-hard, and then we will show how to use or modify the construction for proving the hardness of other variants claimed in Table~\ref{table:1} except $(V,H)$-stabbing. To prove the $\np$-hardness of $(V,H)$-stabbing, we will show a completely different reduction from the Positive Planar Cycle 1-In-3SAT problem \cite{DBLP:conf/wads/ChaplickFLRVW17}. Note that $\np$-hardness of any of these variants does not directly imply the $\np$-hardness of any other variant. 

\subsection{$(H\sqcup V,H\sqcup V)$-stabbing}
We first give some definitions and then we describe the reduction. Consider an instance $I$ of the 3SAT problem with $n$ variables and $m$ clauses. The instance $I$ is called \emph{monotone} if each clause is monotone; that is, each clause consists of only positive literals (called \emph{positive clauses}) or only negative literals (called \emph{negative clauses}). The 3SAT problem is $\np$-hard even when restricted to monotone instances~\cite{DBLP:books/fm/GareyJ79}.

We can associate a bipartite \emph{variable-clause graph} $G_I=(V,E)$ with $I$, where the vertices in one partition of $G_I$ correspond to the variables in $I$ and the vertices in the other partition of $G_I$ correspond to the clauses of $I$. Each clause vertex is adjacent to the variable vertices it contains. The instance $I$ is called \emph{planar} if $G_I$ is planar; it is known that the Planar 3SAT problem is $\np$-hard~\cite{Lichtenstein82}. Moreover, when $I$ is an instance of the Planar 3SAT problem, Knuth and Raghunathan~\cite{DBLP:journals/siamdm/KnuthR92} showed that $G_I$ can be drawn on a grid such that all variable vertices are on a vertical line $\ell$ and clause vertices are connected from left or right of that line in a comb-shaped form without any edge crossing. Moreover, in such a drawing of $G_I$, each variable vertex is drawn as a point in the plane and each clause vertex is drawn as a vertical line segment that is spanned from its lowest variable to its highest variable. More precisely, if a clause $C$ contains three variables $u,v$ and $w$ such that $v$ is between $u$ and $w$ on $\ell$, then the clause vertex is drawn as a vertical line segment $s: x_s\times [y_u,y_w]$. The clause vertex is then connected to its three variable vertices using horizontal line segments. See Figure~\ref{fig:planar3SAT} for an example. In the Planar Monotone 3SAT problem, for any instance $I$, $G_I$ can be drawn as described above. Moreover, all positive clauses (resp., negative clauses) lie to the left (resp., right) of the vertical line $\ell$. De Berg and Khosravi~\cite{DBLP:journals/ijcga/BergK12} showed that Planar Monotone 3SAT is $\np$-hard.

\begin{figure}[t]
\centering
\includegraphics[width=0.4\textwidth]{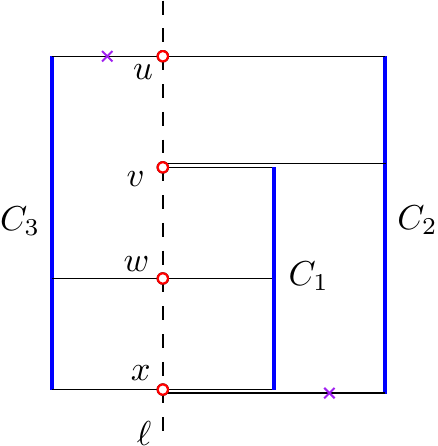}
\caption{An instance of the planar 3SAT in the comb-shaped form of Knuth and Raghunathan~\cite{DBLP:journals/siamdm/KnuthR92}. Crosses on the edges indicate negations; e.g., see $C_2=(u\lor v\lor \overline{x})$.}
\label{fig:planar3SAT}
\end{figure}

\paragraph{Reduction.} We reduce the Planar Monotone 3SAT problem to the $(H\sqcup V,H\sqcup V)$-stabbing problem. For the rest of this section, let $I$ be an instance of the Planar Monotone 3SAT problem with $n$ variables and $m$ clauses. First, we consider a planar monotone drawing of the variable-clause graph $G_I$. 
As mentioned before, this is similar to the non-crossing comb-shaped form of Knuth and Raghunathan~\cite{DBLP:journals/siamdm/KnuthR92}, where (i) variable vertices are all on the vertical line $\ell: x=0$, (ii) the clause vertices are drawn as vertical line segments as described above, and (iii) all the positive clauses (resp., negative clauses) are to the left (resp., right) of $\ell$. Next, we replace each variable vertex $v$ with three horizontal line segments $v_l, v_r$ and $s(v)$. First, $v_l:[x_s,0]\times y_v$ (resp., $v_r:[0,x_{s'}]\times y_v$), where $s$ (resp., $s'$) is the vertical line segment corresponding to the left-most positive (resp., right-most negative) clause that contains $v$. Moreover, $s(v):[-\epsilon,\epsilon]\times y_v$ for some $\epsilon>0$ such that $s(v)$ does not intersect any vertical line segment corresponding to a clause. See Figure~\ref{fig:gadgetOne} for an example. This forms the set of horizontal line segments. Finally, we take the line segments corresponding to clause vertices as the set of vertical line segments. This concludes our instance $I'$ of the $(H\sqcup V,H\sqcup V)$-stabbing problem. Clearly, every horizontal line segment intersects the vertical line $\ell$. Also, there are exactly $3n$ horizontal line segments and $m$ vertical line segments in $I'$, and the instance $I'$ can be constructed in polynomial time.
\begin{lemma}
\label{lem:reductionOne}
The instance $I$ is satisfiable if and only if the instance $I'$ has a feasible solution of size $n$, where $n$ is the number of variables in $I$.
\end{lemma}

\begin{proof}
First, suppose that there exists an assignment that satisfies all clauses in $I$. We construct a feasible solution $S$ for $I'$ as follows. For each variable $v$, if $v$ is set to true (resp., false), then we add the line segment $v_l$ (resp., $v_r$) to $S$. Clearly, $|S|=n$ and every horizontal line segment is stabbed by some line segment in $S$. Now, take any vertical line segment $s\in I'$. If $s$ corresponds to a positive (resp., negative) clause $C$, then at least one of the variables $v$ in $C$ is set to true (resp., false) and so we have added $v_l$ (resp., $v_r$) to $S$. So, every line segment in $I'$ is stabbed by some line segment in $S$.

Now, suppose that there exists a feasible solution $S$ for $I'$ such that $|S|=n$. We assume w.l.o.g. that $s(v)\notin S$ for all variables $v$ of $I$; this is because we can always replace such a line segment $s(v)$ with $v_l$, and still have a feasible solution for $I'$ with the same size $n$. Since $|S|=n$, we must have $v_l\in S$ or $v_r\in S$ for all variables $v$ of $I$, because if there is a variable $v$ for which $v_l\notin S$ and $v_r\notin S$, then no line segment in $S$ can dominate $s(v)$ --- a contradiction to feasibility of $S$. This implies that exactly one of $v_l$ and $v_r$ is in $S$ for all variables $v$ of $I$, and so no vertical line segment can be in $S$, as $|S|=n$. We now obtain a true assignment for $I$ as follows. For each variable $v$, we set $v$ to true (resp., false) if $v_l\in S$ (resp., $v_r\in S$). To see why this is a true assignment for $I$, take any clause $C$ and let $s\in S$ be a line segment that dominates the vertical line segment corresponding to $C$. If $C$ is a positive (resp., negative) clause, then we have set the variable corresponding to $s$ to true (resp., false) and so $C$ is satisfied.
\end{proof}

\begin{figure}[t]
\centering
\includegraphics[width=0.45\textwidth]{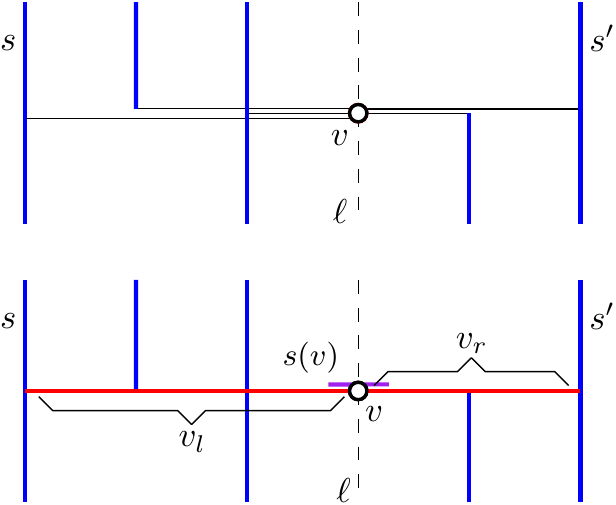}
\caption{A variable $v$ with all the clauses that contain a literal of $v$ (top). The three horizontal line segments $v_l, v_r$ and $s(v)$ corresponding to $v$ (bottom).}
\label{fig:gadgetOne}
\end{figure}

By Lemma~\ref{lem:reductionOne}, we have the following theorem.
\begin{theorem}
\label{thm:firstHardness}
The $(H\sqcup V,H\sqcup V)$-stabbing problem is $\np$-hard.
\end{theorem}

\subsection{Other $\np$-hardness Results}
Consider the construction in Section~\ref{sec:npHardOne}. Observe that the set $S$ in the proof of Lemma~\ref{lem:reductionOne} consists of only horizontal line segments; that is, the problem is $\np$-hard even if we are restricted to selecting a minimum number of horizontal line segments only to stab all line segments.
\begin{theorem}
The $(H,H\sqcup V)$-stabbing problem is $\np$-hard.
\end{theorem}

The construction can be modified to show the hardness of the $(H,V)$-stabbing and $(H\sqcup V,V)$-stabbing problems. To this end, for every variable vertex $v$, we remove the horizontal line segment $s(v)$ and instead add a small vertical line segment $s'(v)$ such that it intersects $v_l$ and $v_r$ only. (To this end, we can e.g. extend $v_r$ slightly to the left of $\ell$ and place $s'(v)$ to the left of $\ell$ and very close to it.) Considering the resulting set of line segments to be an instance of the $(H,V)$-stabbing problem, one can prove a result similar to Lemma~\ref{lem:reductionOne}: the first direction is true as every vertical line segment is stabbed by some horizontal line segment. Moreover, we need to stab the small vertical line segment $s'(v)$, for every variable vertex $v$, which leads to a truth assignment for the instance $I$.

We notice that this also shows the $\np$-hardness of the $(H\sqcup V,V)$-stabbing problem. This is because the vertical line segments are pairwise disjoint; hence, given a feasible solution for the $(H\sqcup V,V)$-stabbing problem, we can obtain a feasible solution of the same size by replacing each vertical line segment by a horizontal one. Hence, we have the following theorem.
\begin{theorem}
\label{thm:HVstabbing}
The $(H,V)$-stabbing problem is $\np$-hard. Moreover, the $(H\sqcup V,V)$-stabbing problem is also $\np$-hard.
\end{theorem}

\subsection{$(V,H)$-stabbing}
Now, we show that the $(V,H)$-stabbing problem is $\np$-hard when each horizontal line segment must be stabbed by exactly one line segment. We show a reduction from the following variant of the 3SAT problem, which was shown to be $\np$-hard by Chaplick et al.~\cite{DBLP:conf/wads/ChaplickFLRVW17}. In an instance $I$ of the Positive Planar Cycle 1-In-3SAT problem, we have $n$ variables and $m$ clauses together with an embedding of $G_I+C$, where $C$ is a cycle through all clause vertices. Moreover, each clause contains exactly three variables and all literals are positive. The problem is to decide whether $I$ is \emph{satisfiable}; here, being satisfiable means whether there exists an assignment of the variables such that exactly one variable in each clause is true. Notice that the problem remains $\np$-hard if we are also given an integer $1\leq k\leq n$ and the problem is to decide if there exists an assignment of the variables such that at most $k$ variables are set to true and exactly one variable in each clause is true.

\paragraph{Reduction.} Given an instance of the Positive Planar Cycle 1-In-3SAT problem, consider the embedding of $G_I+C$. We first transform the embedding into another one in which the cycle $C$ is a half-circle and all the clauses are positioned on the vertical diameter of this half-circle. Let $\ell$ be the vertical line through the diameter. Consider the clause vertices on $\ell$ from top to bottom. Next, we transform the embedding into a ``comb-shaped'' form, where clause vertices are drawn as horizontal line segments and variable vertices are drawn as vertical line segments. We do this transformation in such a way that the vertical line segment corresponding to a vertex $v$ passes through $v$ and spans from the top-most clause to the bottom-most clause that contain it. Moreover, the horizontal line segment corresponding to a clause $C$ passes through $C$ and spans from the left (resp., right) to the left-most (resp., right-most) variable vertex that it contains. This ensures that the horizontal line segment corresponding to a clause $C$ intersects exactly those vertical line segments that correspond to variables that $C$ contain. Notice that this is similar to the ``comb-shaped'' form of Knuth and Raghunathan~\cite{DBLP:journals/siamdm/KnuthR92} except the position of variable and clause vertices are swapped. The set of all vertical and horizontal line segments forms our instance $I'$ of the $(V,H)$-stabbing problem. Observe that $I'$ consists of $n$ vertical and $m$ horizontal line segments, and that it can be constructed in polynomial time.
\begin{lemma}
\label{lem:reductionTwo}
The instance $I$ is satisfiable with $k$ variables set to true if and only if there exists a feasible solution of size $k$ for $I'$.
\end{lemma}
\begin{proof}
First, suppose that $I$ is satisfiable with $k$ variables set to true. Then, we select the $k$ vertical line segments corresponding to these variables as the solution for $I'$. Each horizontal line segment $s$ of $I'$ is stabbed by exactly one selected line segment: the one that satisfies the clause corresponding to $s$. Now, suppose that $I'$ has a feasible solution $S$ of size $k$. Notice that each line segment of $S$ is vertical and so we set to true exactly those variables that correspond to the line segments in $S$. Clearly, $k$ variables are set to true. Moreover, since $S$ is a feasible solution for $I'$, each horizontal line segment is stabbed by exactly one vertical line segment; that is, exactly one variable per clause is set to true.
\end{proof}

By Lemma~\ref{lem:reductionTwo}, we have the following theorem.
\begin{theorem}
The $(V,H)$-stabbing problem is $\np$-hard when every horizontal line segment must be stabbed by exactly one line segment.
\end{theorem}

\section{$(H\sqcup V,H)$-Stabbing}
\label{sec:HUV,H}
In this section, we first design a 5-approximation for $(H\sqcup V,H)$-stabbing that runs in $O(n^6)$ time, and then give a $\ptas$ based on the local search technique. 
\subsection{A 5-approximation for $(H\sqcup V,H)$-stabbing}
Our algorithm is based on a reduction to three simpler problems via LP. The approach is very similar to the one used in \cite{abs-1809-09990}. Recall that $L_v$ is the common vertical line that the horizontal segments intersect. Let $V(l)$ and $V(r)$ be the vertical line segments that lie on the left and right of $L_v$, respectively. For simplicity, suppose $V=V(l)\sqcup V(r)$. Fix an optimum solution $\opt$. Let $H_1$ (resp. $H_2$) be the set of horizontal line segments that get hit by $V(l)$ (resp. $V(r)$) in $\opt$. Note that the line segments in $H_3=H\setminus (H_1\cup H_2)$ get hit by horizontal line segments in $\opt$. If we knew the three sets $H_1,H_2$ and $H_3$, we could solve three subproblems $(V(l),H_1)$-stabbing, $(V(r),H_2\setminus H_1)$-stabbing and $(H,H_3)$-stabbing, and return the union of these three solutions. It is not hard to see that the returned solution has size $|\opt|$. We do not know those three sets, but one can guess (modulo constant approximation) these sets from a fractional LP solution to $(H\sqcup V,H)$-stabbing.

For an LP or ILP $M$, let $\opt(M)$ be its optimum cost. First, we consider the ILP (denoted by ILP1) of $(H\sqcup V,H)$-stabbing. For each line segment $j\in H\sqcup V$, we take a standard 0-1 variable $y_j$ denoting whether $j$ is chosen in the solution or not. For $i\in H$, Let $V(l)^i$ (resp. $V(r)^i$ and $H^i$) be the line segments in $V(l)$ (resp. $V(r)$ and $H$) that stab $i$. Following is the LP relaxation of the ILP.

\begin{mdframed}[backgroundcolor=gray!9]
\begin{alignat*}{3}
& \text{minimize}  && \sum_{j\in H\sqcup V} y_j && \text{(LP1)}\\
& \text{subject to } && \sum_{j\in V(l)^i} y_j+\sum_{j\in V(r)^i} y_j &&\\
& \quad && +\sum_{j\in H^i} y_j \ge 1 \quad  && \forall i\in H\\
& \quad && y_j \in [0,1]  && j \in H\sqcup V
\end{alignat*}
\end{mdframed}

We first solve this LP to obtain a fractional optimal solution $y^*=\{y_j^*:j\in H\sqcup V\}$. Let
\[
H_l=\{i\in H:\sum_{j\in V(l)^i} y_j^*\ge \frac{2}{5}\},
\]
\[
H_r=\{i\in H:\sum_{j\in V(r)^i} y_j^*\ge \frac{2}{5}\},
\]
and
\[
H_h=\{i\in H:\sum_{j\in H^i} y_j^*\ge \frac{1}{5}\}.
\]
Now, consider the following ILPs.

\begin{mdframed}[backgroundcolor=gray!9]
\begin{alignat*}{3}
& \text{minimize}  && \sum_{j\in V(l)} y_j && \text{(ILP2)}\\
& \text{subject to } && \sum_{j\in V(l)} y_j\ge 1 \quad  && \forall i\in H_l\\
& \quad && y_j \in \{0,1\}  && j \in V(l)
\end{alignat*}
\end{mdframed}

\begin{mdframed}[backgroundcolor=gray!9]
\begin{alignat*}{3}
& \text{minimize}  && \sum_{j\in V(r)} y_j && \text{(ILP3)}\\
& \text{subject to } && \sum_{j\in V(r)} y_j\ge 1 \quad  && \forall i\in H_r\\
& \quad && y_j \in \{0,1\}  && j \in V(r)
\end{alignat*}
\end{mdframed}

\begin{mdframed}[backgroundcolor=gray!9]
\begin{alignat*}{3}
& \text{minimize}  && \sum_{j\in H} y_j && \text{(ILP4)}\\
& \text{subject to } && \sum_{j\in H} y_j\ge 1 \quad  && \forall i\in H_h\\
& \quad && y_j \in \{0,1\}  && j \in H
\end{alignat*}
\end{mdframed}

Note that the problems corresponding to ILP2 and ILP3 are precisely the problem of stabbing horizontal rays using vertical line segments \cite{KatzMN05}. Also, ILP4 is corresponding to the stabbing problem of horizontal line segments using horizontal line segments where all line segments intersect $L_v$. By definition of $H_l, H_r$ and $H_h$, there must be a feasible solution for each of these three ILPs. We use the algorithm in \cite{KatzMN05} to obtain optimum solutions $S_1$ and $S_2$ for ILP2 and ILP3, respectively. Also, an optimum solution $S_3$ of ILP4 can be obtained using a simple greedy selection scheme. Finally, we return the solution $S_1\cup S_2\cup S_3$.

It is not hard to see that $S_1\cup S_2\cup S_3$ hits $H$. Next, we argue that $|S_1\cup S_2\cup S_3|\le 5\cdot \opt$. Let $LP2$ (resp. $LP3$ and $LP4$) be the LP relaxation of $ILP2$ (resp. $ILP3$ and $ILP4$). Consider the optimum solution $y^*=\{y_j^*:j\in H\sqcup V\}$ of LP1. First, we have the following simple observation. 

\begin{obs}\label{obs:lp3bound}
$\opt(LP2)$ $\le 2.5\cdot \sum_{j\in V(l)} y_j^*$, $\opt(LP3)$ $\le 2.5\cdot \sum_{j\in V(r)} y_j^*$, and $\opt(LP4)$ $\le 5\cdot \sum_{j\in H} y_j^*$. 
\end{obs}

From the work of \cite{abs-1809-09990}, it follows that the integrality gap of the stabbing problem of horizontal rays using vertical line segments is 2. In particular, their result can be summarized as follows.

\begin{lemma}\label{lem:rayig}
\cite{abs-1809-09990} $\opt(ILP2)\le 2\cdot \opt(LP2)$ and $\opt(ILP3)\le 2\cdot \opt(LP3)$. 
\end{lemma}

We will use Lemma \ref{lem:rayig} to prove our approximation bound. Before that we prove the following lemma.

\begin{lemma}\label{lem:hhig}
$\opt(ILP4)\le \opt(LP4)$. 
\end{lemma}
\begin{proof}
Given a fractional solution $S$ to LP4, we show how to round it to an integral solution $S'$ such that the cost of $S'$ is at most the cost of $S$. First, we partition the set $H$ to a collection of subsets such that each subset contains line segments having the same $y$-coordinates. Then, we round the subsets independent of each other. Consider a particular subset $T$ and a line segment $h\in H_h$ whose $y$-coordinate is the same as that of the line segments in $T$. If there is no such $h$, we set the $y$-values of all line segments in $T$ to 0. Otherwise, as $S$ is a feasible solution, $\sum_{j\in T} y_j\ge 1$. We pick any line segment $h_1$ from $T$ arbitrarily, set its $y$-value to 1, and set the $y$-values of all remaining line segments in $T$ to 0. As $h_1$ hits all line segments in $T$, it is a feasible solution for this subset. Hence, the lemma follows.   
\end{proof}


\begin{lemma}
$|S_1\cup S_2\cup S_3|\le 5\cdot \opt$.
\end{lemma}

\begin{proof}
First note that $|S_1\cup S_2\cup S_3|=\opt(ILP2)+\opt(ILP3)+\opt(ILP4)$. From Lemmas \ref{lem:rayig} and \ref{lem:hhig}, it follows that $|S_1\cup S_2\cup S_3|\le 2\cdot\opt(LP2)+2\cdot\opt(LP3)+\opt(LP4)$. Then, as the sets $V(l), V(r)$ and $H$ are pairwise disjoint, by Observation \ref{obs:lp3bound}, $|S_1\cup S_2\cup S_3|\le 5\cdot \opt(LP1)\le 5\cdot \opt$. 
\end{proof}

To solve LP1 one can use the LP solver in \cite{Tardos86} that runs in $O(n^5)$ time. The algorithm to compute $S_1,S_2$ and $S_3$ takes $O(n^6)$ time in total \cite{KatzMN05}. Hence, our approximation algorithm runs in $O(n^6)$ time. 

\begin{theorem}
There is a 5-approximation for $(H\sqcup V,H)$-stabbing that runs in $O(n^6)$ time.  
\end{theorem}

\subsection{A $\ptas$ for $(H\sqcup V,H)$-stabbing}
\label{sec:ptas}
We would like to hit all the line segments of $H$ using a minimum size subset of $H\sqcup V$. We design a local search based $\ptas$ for this problem. Note that if a horizontal line segment $h_1$ in a solution hits a horizontal line segment, then they must have the same $y$-coordinates, and $h_1$ hits all the line segments that have the same $y$-coordinates as $h_1$. Thus, for a subset $S$ of horizontal line segments with the same $y$-coordinates, we can identify one of them (say $h$) and assume that if the line segments in $S$ are being hit by a horizontal line segment, then that line segment is $h$. We preprocess the set $H$ to compute a subset $H'$ so that for each subset $S$ of horizontal line segments with the same $y$-coordinates, $|H'\cap S|=1$. Thus, now we would like to find a minimum size subset of  $H'\sqcup V$ that hit all the line segments in $H$.

We use a standard local search algorithm as the one in \cite{BandyapadhyayM018}, which runs in $n^{O(\frac{1}{\epsilon^2})}$ time. Let $\ml$ (blue) and $\g$ (red) be the local and global optimum solutions, respectively and w.l.o.g., assume $\ml\cap \g =\emptyset$. Then, by Theorem \ref{thm:LSgivesPTAS}, to say that the local search algorithm is a $\ptas$ it is sufficient to prove the existence of a bipartite local exchange graph $G=(\ml, \g,E)$ that is also planar. Note that $G=(\ml,\g,E)$ is a local exchange graph if for each $h\in H$, there exists $s_1\in \ml, s_2\in \g$ such that $s_1\cap s_2\cap h\ne \emptyset$ and $(s_1,s_2)\in E$. We will construct a plane graph $G=(\ml,\g,E)$ in the following, which satisfies the local exchange property. 

For each line segment $h \in H'\cap (\ml\cup \g)$, we select the intersection point $h\cap L_v$ to draw the vertex for $h$. For each line segment $v\in V\cap (\ml\cup \g)$, we select $v$ itself to draw the vertex for $v$. Later we will contract each such $v$ to a single point. For each red (resp. blue) $h\in H'\cap \g$ (resp. $h\in H'\cap \ml$), let $v_1$ and $v_2$ be the first line segments of color blue (resp. red) on the left and right of $L_v$, respectively that intersect the line-extension $L(h)$. Note that $v_1$ and $v_2$ might not exist. We add two edges $(h,v_1)$ and $(h,v_2)$ (or horizontal line segments), between $h\cap L_v$ and $I(L(h),v_1)$ and between $h\cap L_v$ and $I(L(h),v_2)$, respectively. For each  $h\in H'\setminus (\ml\cup \g)$, let $V'$ be the line segments in $\g\cup \ml$ that intersect $L(h)$. We add an edge between each consecutive (w.r.t. $V'$) pair of line segments $(v,v')$ such that $v\in \g$ and $v'\in \ml$. In particular, we draw a horizontal line segment between $I(L(h),v)$ and $I(L(h),v')$. 

Note that each edge of $G$ is a horizontal line segment and any pair of those can be drawn in a non-overlapping manner. Also, an input vertical segment can intersect any such edge only at one of its vertices. Thus, the planarity of the graph follows. Lastly, we contract each vertical line segment $v$ to a point, which does not violate the planarity. Now, consider any line segment $h_1\in H$. Suppose $h_1$ is hit by a horizontal line segment $h\in \g\cup \ml$. W.l.o.g., let $h\in \g$. Let $v_1$ and $v_2$ be the first line segments of color blue on the left and right of $L_v$, respectively that intersect $L(h)$. Then, either $v_1$ or $v_2$ must hit $h_1$, as $h_1$ intersects the line $L_v$. As we add the edges $(h,v_1)$ and $(h,v_2)$, the local exchange property holds for $h_1$. Now, suppose $h_1$ does not get hit by a horizontal line segment in $\g\cup \ml$. Let $V'$ be the line segments in $\g\cup \ml$ that intersect $L(h_1)$. Then, there must be two consecutive (w.r.t. $V'$) vertical line segments $v\in \g$ and $v'\in \ml$ both of which hit $h_1$. As we add the edge $(v,v')$, the local exchange property holds for $h_1$ in this case as well. It follows that the local search algorithm is a $\ptas$, and we obtain the following theorem. 

\begin{theorem}\label{thm:hvandh}
There is a $(1+\epsilon)$-approximation for $(H\sqcup V,H)$-stabbing that runs in $n^{O(\frac{1}{\epsilon^2})}$ time for any $\epsilon > 0$.  
\end{theorem}

\paragraph{Remark.} One can show that $(H\sqcup V,H)$-stabbing is a special case of $(V,H)$-stabbing. Construct an instance $I'$ of $(V,H)$-stabbing from any given instance $I$ of $(H\sqcup V,H)$-stabbing by taking the vertical and horizontal segments in $I$ along with some special vertical segments. For each maximal cluster of horizontal segments having same $y$-coordinates, add one special vertical segment such that the only segments it intersects are the segments in the cluster. Then, given a solution for one instance, a solution for the other of the same size can be computed in polynomial time. As $(V,H)$-stabbing admits a PTAS \cite{BandyapadhyayR17}, one can also obtain a PTAS for $(H\sqcup V,H)$-stabbing using this alternative approach.

\section{$(H\sqcup V,V)$-Stabbing and $(H\sqcup V,H\sqcup V)$-Stabbing}
\label{sec:approx}
In this section, we obtain a 2-approximation for $(H\sqcup V,V)$-stabbing, and then show how this algorithm along with the algorithms of the previous section can be used to obtain approximations for $(H\sqcup V,H\sqcup V)$-stabbing.

\subsection{A 2-approximation for $(H\sqcup V,V)$-stabbing}
First, we assume that the line segments of $V$ are lying only on the right side of $L_v$. We design a polynomial time algorithm for this case, and later we will show how to obtain a 2-approximation for the general case of $(H\sqcup V,V)$-stabbing by using this algorithm as a subroutine. Now, as the line segments of $V$ are lying only on the right side of $L_v$, w.l.o.g., we can assume that the $y$-coordinates of the line segments in $H$ are distinct. Also, for simplicity, we assume that the endpoints of all the line segments are distinct.

Our algorithm is based on dynamic programming. Let $h_1,h_2,\ldots,h_t$ be the line segments in $H$ in increasing order of their $y$-coordinates. Also, let $v_1,v_2,\ldots,v_m$ be the line segments in $V$ in non-decreasing order of their $x$-coordinates. If two vertical line segments have the same $x$-coordinates, order them in decreasing order of $y$-coordinates of their bottom endpoints. Now, we describe the subproblems we consider. Each subproblem $(i,j,k,k')$ is defined by four indexes $i$, $j$, $k$ and $k'$, where $1\le j\le i \le t$, and $1\le k,k' \le m$. Let $H(i,j)=\{h_i,h_{i-1},\ldots,h_j\}$.
Also, let $V_l=\{v_1,v_2,\ldots,v_l\}$, and $V(i,j,l)$ be the line segments of $V_l$ each of which intersects at least one line segment in $H(i,j)$ and does not intersect $h_{i+1}$ or $h_{j-1}$ (if exists). In subproblem $(i,j,k,k')$, we would like to hit $V(i,j,k')$ using a minimum size subset of $H(i,j)\sqcup V_k$ (see Figure \ref{fig:dp}). We define $f(i,j,k,k')$ to be the size of an optimum solution of the subproblem $(i,j,k,k')$. Note that we are interested in computing $f(t,1,m,m)$.

\begin{figure}
\centering
\includegraphics[width=0.55\textwidth]{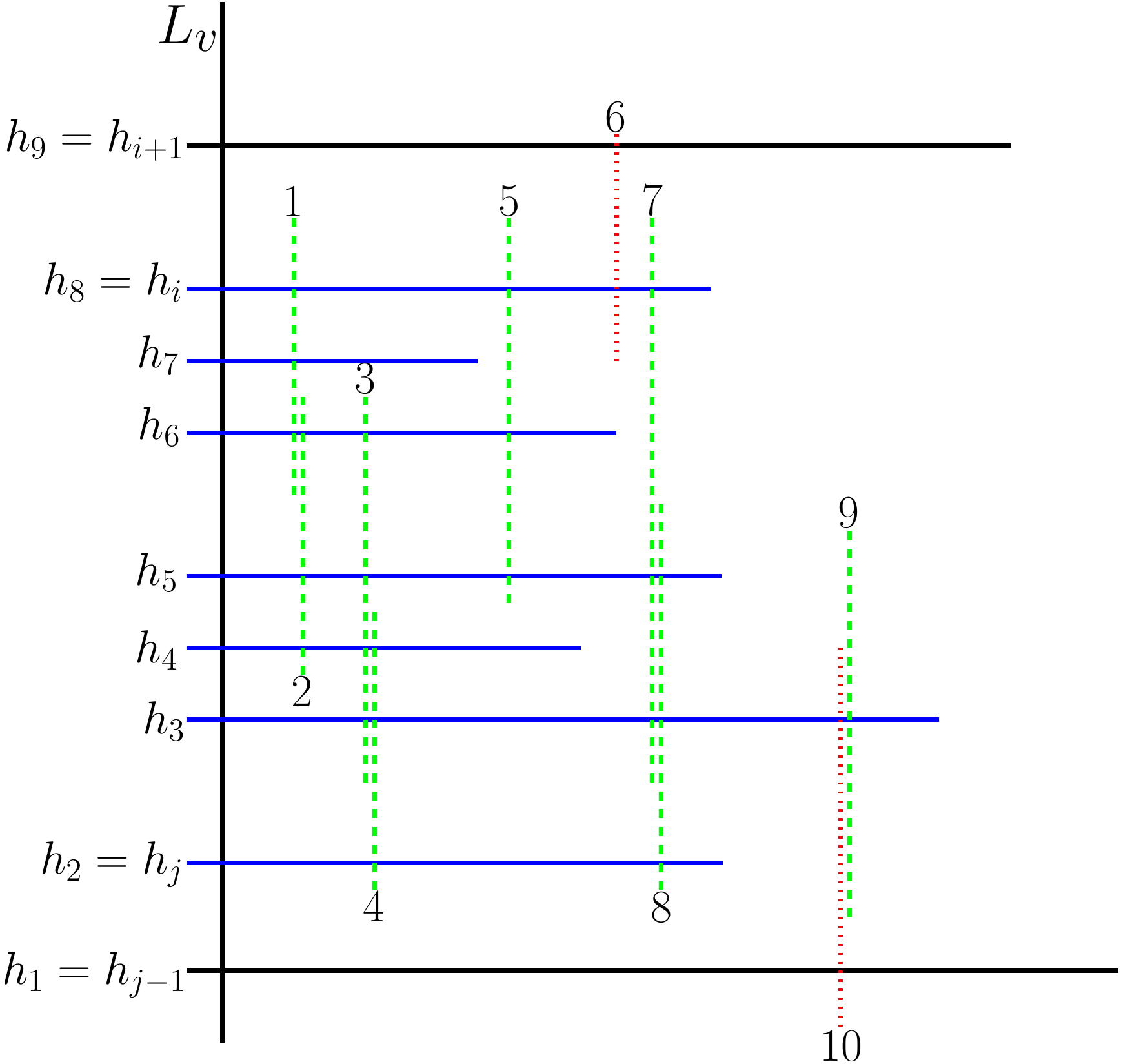}
\caption{An example demonstrating a subproblem $(i,j,k,k')$ with $i= 8, j= 2, k= 10$ and $k'=9$. The line segments in $V(i,j,k')$ are shown using dashed (green) line segments. The dotted (red) line segments either intersect $h_{i+1}$ or has index more than 9, and thus are not in $V(i,j,k')$.}
\label{fig:dp}
\end{figure}

We use the following recursive structure to compute $f(i,j,k,k')$. Let $v$ be the line segment in $V(i,j,k')$ having the maximum index. Now, there can be two cases. $v$ gets hit by a horizontal line segment in optimum solution or $v$ gets hit by only vertical line segments in the optimum solution. For the first case, we guess a line segment $h_{i'}$ that hits $v$. Let $k_1'$ be the maximum index of the vertical line segments in $V(i,j,k')$ that intersect at least one line segment in $H(i,i'+1)$ and do not intersect $h_{i+1}$ or $h_{i'}$. Similarly, let $k_2'$ be the maximum index of the vertical line segments in $V(i,j,k')$ that intersect at least one line segment in $H(i'-1,j)$ and do not intersect $h_{i'}$ or $h_{j-1}$. Then, the two new subproblems that we need to solve are $(i,i'+1,k,k_1')$ and $(i'-1,j,k,k_2')$. For the second case, we guess the maximum index $l$ of the line segments in $V_k$ that hit $v$ in the optimum solution. Let $k_3'$ be the maximum index of the vertical line segments in $V(i,j,k')$ that does not intersect $v_l$. From the definition of $v$ and the fact that $x$-coordinates of $v$ and $v_l$ are same, it follows that there is no line segment in $V(i,j,k')$ with index larger than $k_3'$ that does not intersect $v_l$. Thus, the new subproblem that we need to solve is $(i,j,l-1,k_3')$. 
Note that the number of guesses for $h_{i'}$ and $v_l$ is $O(n)$. Thus, to solve $(i,j,k,k')$ the total number of subproblems that we need to solve is $O(n)$. Using a dynamic programming based scheme these subproblems can be evaluated easily. Since the number of distinct subproblems $(i,j,k,k')$ is $O(n^4)$, all the subproblems can be solved in $O(n^5)$ time.

\paragraph{$(H\sqcup V,V)$-stabbing.} 
Let $V(l)$ and $V(r)$ be the vertical line segments that lie on the left and right of $L_v$, respectively. Set $S=H\sqcup V$, $D_1=V(l)$ and $D_2=V(r)$. Then, note that the problem $(S,D_1)$-stabbing is same as $(H\sqcup V(l),D_1)$-stabbing. Similarly, $(S,D_2)$-stabbing is same as $(H\sqcup V(r),D_2)$-stabbing. Thus, we can use the above algorithm to solve these two problems, and hence there exists exact algorithms for $(S,D_1)$-stabbing and $(S,D_2)$-stabbing. From Observation \ref{obs:merge}, we obtain the following theorem.
\begin{theorem}
There is a 2-approximation for $(H\sqcup V,V)$-stabbing that runs in $O(n^5)$ time.
\end{theorem}

\paragraph{$(H\sqcup V,H\sqcup V)$-stabbing.}
Set $S=H\sqcup V$, $D_1=H$ and $D_2=V$. Then, we know that there are two algorithms for the $(S,D_1)$-stabbing problem: an $O(n^6)$-time 5-approximation algorithm and a $\ptas$. Moreover, we have an $O(n^5)$-time 2-approximation algorithm for the $(S,D_2)$-stabbing problem. Therefore, by Observation~\ref{obs:merge}, we have the following theorem.
\begin{theorem}
There is a 7-approximation for $(H\sqcup V,H\sqcup V)$-stabbing that runs in $O(n^6)$ time and a $(3+\epsilon)$-approximation that runs in $n^{O(\frac{1}{\epsilon^2})}$ time for any $\epsilon > 0$.  
\end{theorem}

\bibliographystyle{plain}
\bibliography{ref}

\end{document}